\documentclass[10pt]{article}
\usepackage{amsfonts,amsmath,amsthm,epsfig,verbatim,amssymb,amscd,eucal,bbm}

\def\Blackboardfont{\mathbb}

\newtheorem{theorem}{Theorem}[section]
\newtheorem{proposition}[theorem]{Proposition}

\newtheorem{pr-def}[theorem]{Proposition-Definition}
\newtheorem{corollary}[theorem]{Corollary}
\newtheorem{lemma}[theorem]{Lemma}
\newtheorem{remark}[theorem]{Remark}

\newcommand{\1}{\mathbbm{1}}
\newcommand{\mrm}[1]{\text{\rm #1}}

\def\eref#1{(\ref{#1})}
\def\Blackboardfont{\mathbb}

\newcommand{\cA}{{\cal A}}
\newcommand{\cD}{{\cal D}}

\newcommand{\cF}{{\cal F}}

\newcommand{\cT}{{\cal T}}
\newcommand{\cR}{{\cal R}}
\newcommand{\cN}{{\cal N}}
\newcommand{\cS}{{\cal S}}

\def\Z{{\Blackboardfont Z}}

\def\N{{\Blackboardfont N}}
\def\R{{\Blackboardfont R}}

\def\triangleup{\vartriangle}
\title{Queues, Stores, and Tableaux}
\author{Moez Draief
 \thanks{LIAFA, Universit\'e Paris 7, case 7014,
2 place Jussieu, 75251 Paris Cedex 05, France.
{\tt draief@liafa.jussieu.fr}.}
\and Jean Mairesse
\thanks{LIAFA, Universit\'e Paris 7, case 7014,
2 place Jussieu, 75251 Paris Cedex 05, France.
{\tt mairesse@liafa.jussieu.fr}.}
\and
Neil O'Connell
\thanks{Mathematics Institute, Univ. Warwick, Coventry CV4 7AL, United
 Kingdom. {\tt noc@maths.warwick.ac.uk}.}
}

\begin{document}

\maketitle

\noindent

\begin{abstract}
Consider the single server queue with an infinite buffer and
a FIFO discipline, either of type $M/M/1$ or Geom/Geom/1.
Denote by  $\cA$ the arrival process and by $s$ the services.
Assume the stability condition to be
satisfied. Denote by $\cD$ the departure process in equilibrium and by
$r$ the time spent by the customers at the very back
of the queue. We prove that $(\cD,r)$ has the same law as $(\cA,s)$
which is an extension of the classical Burke Theorem.
In fact, $r$ can be viewed as the departures from a dual {\em
  storage} model. This duality between the two models also appears
when studying the transient behavior of a tandem by means of the
RSK algorithm: the first and last row of the resulting
semi-standard Young tableau are respectively the last instant of
departure in the queue and the total number of departures in the
store.

\medskip

\end{abstract}

\section{Introduction}

The main purpose of this paper is to clarify the interplay between two
models of queueing theory. The first model is the very classical single
server queue with an infinite buffer and a FIFO discipline. The
second, less common but very natural, model can be described as a queue operating in slotted time
with batch arrivals and services. It was studied for instance in
\cite{BeAz}. 
The literature on discrete-time queues is relatively less furnished
than its continuous-time counterpart. The interest for discrete-time
queues is however gaining ground due to their relevance in ATM-like
communication systems~\cite{BrKi,taka93}. 

In this paper, to clearly distinguish between the two models, we choose
to describe the second one with a different terminology and as a
storage model.

We prove first that the two models are linked in a very strong
way. We set up an abstract model with an ordered pair of input
variables $(\cA,s)$ and an ordered pair of output variables
$(\cD,r)=\Phi(\cA,s)=(\Phi_1(\cA,s),\Phi_2(\cA,s))$. On the one
hand, the queueing model corresponds to $\cD=\Phi_1(\cA,s)$ with
$\cA$ and $s$ being respectively the arrivals and services and
$\cD$ the departures. On the other hand, the storage model
corresponds to $r=\Phi_2(\cA,s)$, with $s$ and $\cA$ being
respectively the supplies (arrivals) and requests (services) and
$r$ the departures. The interpretation of either $r$ for the queue
or $\cD$ for the store is much less natural.

Then we assume that the random variables driving the dynamic are
either exponentially or geometrically distributed, and we consider
the models in equilibrium (under the stability condition). In this
situation, it is well known that a Burke's type Theorem holds: the
departures and the arrivals have the same law
\cite{burk,reic,BeAz}. This can be considered as one of the
cornerstones of queueing theory, see for instance the books
\cite{brem,kell79,robe03} for discussions and related materials.
Here we prove a `joint' version of the result: $(\cD,r)$ has the
same law as $(\cA,s)$. This joint Burke Theorem is new, although
it is similar in spirit to the results proved in \cite{OY02,KOR}
for a variant: queues with {\em unused services}. Also, in the
geometric case, we use an original method of proof based on the
reversibility of a symmetrized version of the workload process
(instead of the queue length process). As in \cite{OY02,KOR}, the
joint Burke's Theorem can be used to obtain a representation
theorem for the joint law of two independent random walks with
exponential or geometric jumps conditioned on never colliding.

A second facet of the duality between queues and stores appears
when considering the Robinson-Schensted-Knuth (RSK) algorithm. 
The RSK algorithm
is an important object in representation theory and symmetric functions
theory and it has been known for some time now that is has connections
with queueing theory, although these connections have remained somewhat 
mysterious until now.  In this paper we elucidate these connections
completely via the above duality.  The results we obtain, apart from
being interesting in their own right, can be used as a starting point
for the application of symmetric functions theory (and the associated
analysis) to obtain interchangeability as well as asymptotic results 
for models which are constructed by putting the above queues and stores 
in series.  These results are also applicable in the context of the
corresponding interacting particle systems, where there is much interest
in obtaining exactly solvable microscopic models for Burger-type
equations. 

Let us now detail the results. Consider $K$
queues (stores) in tandem: the departures from a queue (store) are
the arrivals (supplies) at the next one. Initially,
the network is empty except for an infinite number
of customers (infinite supply) in the first queue (store).

The particle systems associated with the above tandems 
are described in Section \ref{sse-part}. One model is original and
worth putting forward: the {\em bus-stop model} which is the
zero-range model associated with stores in tandem. 

Here we assume that the variables driving the dynamic are
$\N$-valued without any further assumption.
Building on ideas developed in \cite{bary01,ocon03},
we can study the transient evolution as follows. Consider the
family of r.v.'s $(u(i,j))_{1\leq i \leq N, 1 \leq j \leq K}$
where $u(i,j)$ is the service of the $i$-th customer at queue $j$,
resp. the request at time slot $i$ in store $K+1-j$. Apply the RSK
algorithm, see \cite{knut70,saga,stan2}, to this family and let
$P$ be the resulting semi-standard Young tableau (here we do not
consider the recording tableau $Q$). Let $\lambda_1 \geq \cdots
\geq \lambda_K \geq 0$ be the lengths of the successive rows of
$P$. Classically \cite{knut70}, we have $\lambda_1=\max_{\pi \in
\Pi} \sum_{(i,j)\in \pi} u(i,j)$, where $\Pi$ is the set of paths
in $\N^2$ going from $(1,1)$ to $(N,K)$ and which are increasing
and consist of adjacent points. Moreover, it is well known in
queueing theory \cite{GlWh,muth,TeWo,SzKe} that $\max_{\pi \in \Pi}
\sum_{(i,j)\in \pi} u(i,j)$ is $D$: the instant of departure of
customer $N$ from queue $K$. Pasting the two results together
gives the folklore observation that $\lambda_1=D$. Here we
complete the picture by proving that $\lambda_K=R$, where $R$ is
the total number of departures from the last store in the tandem
up to time slot $N$. Again, this identity is proved in two steps
by showing that $\lambda_K=\min_{\pi \in \widetilde{\Pi}}
\sum_{(i,j)\in \pi} u(i,j)=R$, where $\widetilde{\Pi}$ is a
different set of paths in the lattice. To summarize, we obtain on
the same Young tableau the total departures for the two tandem
models.

A short version without proofs of the paper appears in the proceedings
of the conference {\em Discrete Random Walks 03} \cite{DMOc}.

\section{Notations}

We work on a probability
space $(\Omega, \cF,P)$. The indicator function of an event $E\in
\cF$ is denoted by $\1_{E}$. We use the symbol $\sim$ to denote the
equality in distribution of random variables.
Depending on the context, $|A|$ is the cardinal of the set $A$ or the
length of the word $A$. We set $\N^*=\N\setminus \{0\}$,
$\R^*_+=\R_+\setminus \{0\}$, and $x^+=x\vee 0
=\max(x,0)$. We use the convention that $\sum_{i=j}^k u_i =0$ when
$j>k$.

Below, a {\em  point process} is a stochastic simple point process on
$\R$ with an infinite
number of positive and negative points. We identify a point process $\cA$
with the random ordered sequence of its points: $\cA = (A_n)_{n\in
  \Z}$ with $A_n<A_{n+1}$ for all $n$. Observe that the numbering of
the points is defined up to a translation in the indices.
%% It is also
%% useful to have a canonical numbering: we set
%% $\cA=(\cT_n(\cA))_{n\in \Z}$ with $\forall n,
%% \cT_n(\cA)<\cT_{n+1}(\cA),$ and $\cT_0(\cA) \leq 0 < \cT_1(\cA)$.
For any interval $I$, we define the {\em counting} random variable:
$\cA(I)=\sum_{n\in \Z} \1_{\{A_n \in I\}}$.

A {\em marked point process} is a couple $(\cA,c)=(A_n,c_n)_{n\in \Z}$
where $\cA=(A_n)_{n\in \Z}$ is a  point process and $c=(c_n)_{n\in \Z}$
is a sequence of r.v.'s valued in some state space. The {\em mark}
$c_n$ is {\em associated} to the point $A_n$ of the point process. For
precisions concerning point processes see \cite{DaVe}.

\medskip

Given a point process $\cA=(A_n)_{n\in\Z}$, the {\em reversed}
point process $\cR(\cA)$ is the point process obtained by
reversing the direction of time; i.e. $\cR(\cA)=(-A_{-n})_{n\in
\Z}$. Given a marked point process $(\cA,c)=(A_n,c_n)_{n\in \Z}$,
the {\em reversed} marked point process is $\cR(\cA,c)=
(-A_{-n},c_{-n})_{n\in \Z}$.

Given a c\`adl\`ag, i.e. right-continuous and left-limited, random
process $Y=(Y(t))_{t\in \R}$ valued in $\R$, we define the {\em
reversed} process $\cR\circ Y=(\cR\circ Y(t))_{t\in \R}$ as the
c\`adl\`ag modification of the process $(Y(-t))_{t\in \R}$. Denote
by $\cN_+(Y)$ and $\cN_-(Y)$ the point processes (with a possibly
finite number of points) corresponding
respectively to the positive and negative jumps of $Y$, that is
for any interval $I$ of $\R$,
\begin{equation}\label{jumps}
\cN_+(Y)(I)= \int_{I} \1_{\{Y(u)>Y(u-)\}} du, \; \cN_-(Y)(I)=
\int_{I} \1_{\{Y(u)<Y(u-)\}} du\:.
\end{equation}

\section{The Model}
\label{se-model}

Let $\cA=(A_n)_{n \in \Z}$ be a point process and assume that  $A_0\leq 0
 <A_1$.
We define
the $\R_+^*$-valued sequence of r.v's $a=(a_n)_{n \in \Z}$ by $a_n =A_{n+1}
-A_{n}$. Let $s=(s_n)_{n \in \Z}$ be another $\R_+^*$-valued
sequence of r.v's. The marked point process $(\cA,s)$ is the {\em input} of
 the
model.

\medskip

Define the sequence of r.v.'s $\cD=(D_n)_{n \in \Z}$ by
\begin{equation}
\label{Depar.times} D_n = \sup_{k \leq n}\: \Bigl[A_k +\sum_{i=k}^n
s_i \Bigr]\:.
\end{equation}

A priori the $D_n$'s are valued in $\R \cup \{+\infty\}$. Assume
from now on that $(\cA,s)$ is such that the $D_n$'s are almost surely
finite. They
satisfy the recursive equations:
\begin{equation}
\label{eq-recd} D_{n+1} = \max (D_n , A_{n+1}) + s_{n+1}\:.
\end{equation}
Set $\cD=(D_n)_{n\in\Z}$ and set $d_n=D_{n+1}-D_n$. Define
an additional sequence of r.v.'s $r=(r_n)_{n \in \Z}$, valued in
$\R_+^*$, by
\begin{eqnarray}\label{reverserv.times}
r_n &=& \min(D_n, A_{n+1})-A_n\:.
\end{eqnarray}
The marked point process $(\cD,r)$ is the {\em output} of the
model. By summing \eref{eq-recd} and \eref{reverserv.times}, we get
the following interesting relation between input and output variables:
\begin{equation}\label{eq-rel0}
r_n+d_n = a_n +s_{n+1}\:.
\end{equation}

In view of the future analysis, it is convenient to define the
following auxiliary variables. Let $w=(w_n)_{n \in \Z}$ be the
sequence of r.v.'s valued in $\R_+$ and defined by
\begin{equation}
\label{Devolindley} w_n = D_n -s_n -A_n = \sup_{k\leq n-1}
\Bigl[\ \sum_{i=k} ^{n-1} (s_i-a_i) \Bigr]^+ \:.
\end{equation}
These r.v.'s satisfy the recursive equations:
\begin{equation}
\label{Reclindley} w_{n+1} = [w_n + s_n -a_n]^+\:.
\end{equation}
Using the variables $w_n$, we can give alternative definitions of $D_n$
and $r_n$:
\begin{equation}
\label{backD} \forall l\leq n, \ D_n= \Bigl[w_l + A_l + \sum_{i=l}^n s_i \Bigr]
\vee \max_{l< k \leq n}\: \Bigl[A_k +\sum_{i=k}^n
s_i \Bigr]\:,
\end{equation}
\begin{equation}
\label{back2} r_n= \min\{w_n + s_n, a_n\}=s_n+w_n-w_{n+1}\:.
\end{equation}
At last, we define the c\`adl\`ag random process $Q=(Q(t))_{t \in \R}$,
valued in $\N$, by
\begin{equation}
\label{Queue-length} Q(t) = \sum_{n \in \Z} \1_{\{A_n \leq t <
D_n\}}\:.
\end{equation}

\begin{lemma}\label{numbering}
We have $\cN_+(Q)=\cA$ and
$\cN_-(Q)=\cD$. Furthermore, $D_{-Q(0)}\leq 0 < D_{-Q(0)+1}$.
\end{lemma}

\begin{proof}
 Recall that $Q(0) = \sum_{n \in \Z} \1_{\{A_n \leq 0 < D_n\}}$ and
that $A_0 \leq 0 < A_1$. Hence, $Q(0)=|k|$ where $k = \sup\{n \in
\Z_-\:|\: D_n \leq 0\}$. The result follows.
\end{proof}

% The shift in the indices between $(D_n)_n$ and
% $\cT_n\circ\cN_-(Q)$ is illustrated in Figure \ref{figure2}.
%\begin{figure}[hbt]
% \[\epsfxsize=300pt \epsfbox{numbering.eps} \]
% \caption{Numbering of $(D_n)_n$} \label{figure2}
%\end{figure}

We now interpret the variables defined above in two different
contexts: a queueing model and a storage model.

\subsection{The single-server queue}\label{sse-queue}

A bi-infinite string of customers is served at a queueing facility
with a single server. Each customer is characterized by an instant
of arrival in the queue and a service demand. Customers are served
upon arrival in the queue and in their order of arrival. Since
there is a single server, a customer may have to wait in a buffer
before the beginning of its service. Using Kendall's nomenclature,
our model is a $././1/\infty/$FIFO queue.

\begin{figure}[hbt]
\[\epsfxsize=300pt \epsfbox{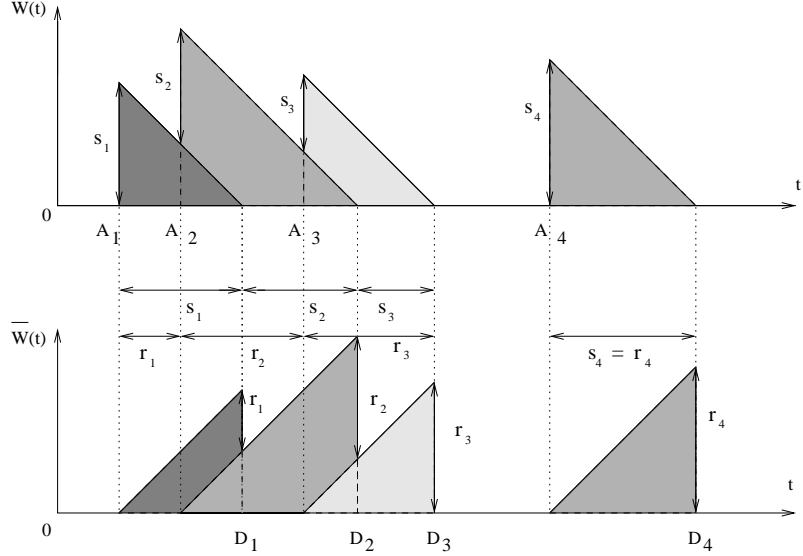} \]
\caption{The dual variables $(s_n)_n$ and $(r_n)_n$}
\label{figure1}
\end{figure}

The customers are numbered by $\Z$ according to their order of
arrival in the queue (customer 1 being the first one to arrive
strictly after instant 0).
Let $A_n$ be the instant of {\bf
A}rrival of customer $n$ and $s_n$ its {\bf S}ervice time. Then
the variables defined in (\ref{Depar.times})-(\ref{Queue-length})
have the following interpretations:
\begin{itemize}
\item $D_n$ is the instant of departure of customer $n$ from the
queue, after completion of its service;
%% $(d_n)_n$ is the sequence of inter-departure times;
\item $w_n$ is the waiting time of customer $n$ in the buffer
between its arrival and the beginning of its service;
\item $Q(t)$ is the number of customers in the
queue at instant $t$ (either in the buffer or in service);
$Q=(Q(t))_t$ is called the {\em queue-length} process;
\item $r_n$ is the time
spent by customer $n$ at the {\em very back} of the queue.
\end{itemize}

The variables $(r_n)_n$ are less classical in queueing theory
although they have already been considered \cite{PrGa}. They
should be viewed as being dual to the services $(s_n)_n$ as
illustrated in Figure \ref{figure1}. On the upper part of the
figure, we have represented the {\em workload} process $(W(t))_t$,
where $W(t)$ is the waiting time of a virtual customer arriving at
instant $t$ (see \eref{eq-workload} for the formal definition).

\subsection{The storage model}\label{sse-storage}

Some product $P$ is supplied, sold and stocked in a store in the following
way. Events occur at integer-valued epochs, called {\em slots}. At
each slot, an amount of $P$ is supplied and an amount of $P$ is
asked for by potential buyers. The rule is to meet all the demand,
if possible. The demand of a given slot which is not met {\em is
lost}. The supply of a given slot which is not sold {\em is
not lost} and is stocked for future consideration.

\medskip

Let $s_n$ be the amount of $P$ {\bf S}upplied at slot $n+1$, and let $a_n$
be the amount of $P$ {\bf A}sked for at the same slot. The variables in
(\ref{Depar.times})-(\ref{Queue-length}) can be interpreted in
this context:
\begin{itemize}
\item $w_n$ is the level of the stock at the end of slot $n$. It
evolves according to (\ref{Reclindley});
\item $r_n$ is the demand met at slot $n+1$, see equation (\ref{back2}); it
is the amount of $P$ departing at slot $n+1$;
\end{itemize}
The variables $(D_n)_n$ and the function $Q$ do not have a natural
interpretation in this model.

\medskip

The evolution of the store is summarized in Figure \ref{fi-store}.
The indices may seem weird ($s_n,a_n,r_n$, for time slot $n+1$).
They were chosen that way to get better looking formulas in \S
\ref{se-trans}.

\begin{figure}[hbt]
\[\epsfxsize=160pt \epsfbox{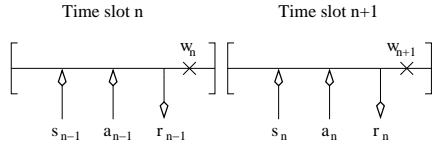} \]
\caption{The storage model}
\label{fi-store}
\end{figure}

It is important to remark that while the equations driving the
{\em single server queue} and the {\em storage model} are exactly
the same, it is not the same variables that make sense in the two
models. The important variables are the ones corresponding to the
departures from the system. The departures are coded in the
variables $ (D_n)_n$ for the single server queue and in the
variables $(r_n)_n$ for the storage model. On the other hand,
interpreting the variables $(r_n)_n$ in the single server queue or
the variables $(D_n)_n$ in the storage model is not so immediate.

\medskip

To summarize, the output variables $ (D_n)_n$ and $(r_n)_n$ are
equally relevant, but in different modelling contexts. 

\medskip

Observe that it would be possible to describe the above storage model
with a ``queueing'' terminology (a queue with slotted time, batch
arrivals ($s_n$)
and batch services ($a_n$)). It is the description used in \cite{BeAz} for
instance.  We have avoided it on purpose to clearly separate the
models of \S \ref{sse-queue} and \S \ref{sse-storage} and therefore
minimize the possibility of confusion.

\section{Equilibrium Behavior: Around Burke's Output Theorem}

We consider the exponential or geometric version of our model in
equilibrium. We prove a Burke's type result: $(\cD,r)\sim
(\cA,s)$. The relevant result for the sequence of departures in
one of the two models will then follow by forgetting one of the
two variables. In the exponential case, our proof uses the
queue-length process $Q$ following \cite{reic}. Our method of
proof in the geometric case is original and based on the zigzag
process, a symmetrized version of the workload process. A
discussion of the literature and of the increment of the present
version is carried out in \S \ref{sse-comm}.

\subsection{Output theorem in the exponential case} \label{sse-exp}

Let $\cA$ be an homogeneous Poisson process of intensity
$\lambda\in \R_+^*$. We set $ \cA=(A_n)_{n\in \Z}$ with
$A_0<0<A_1$. Recall that $a_n=A_{n+1}-A_n$. Then $(a_n)_{n\geq 1}$
is a sequence of i.i.d. r.v.'s with exponential distribution of
parameter $\lambda$. Let $s=(s_n)_{n\in \Z}$ be a sequence of
i.i.d. random variables, independent of $\cA$, with exponential
distribution of parameter $\mu$. We assume that $\lambda < \mu$.

\medskip

Consider the marked point process $(\cA,s)$ as being the
input of the model of Section \ref{se-model}. The sequence
$(w_n)_n$ is a random walk valued in $\R_+$ with a
barrier at $0$. Under the {\em stability} condition $\lambda <
\mu$, this random walk has a negative drift. It implies that the
random variables $D_n$ defined in
(\ref{Depar.times}) are indeed almost surely
finite. 

\medskip

We now prove Theorem \ref{burke1}. 
Figure \ref{fi-store} provides a visual illustration of the
result. The principle of the proof is the same as in Reich~\cite{reic}.
However, the adaptation is not completely immediate. Some care is
necessary in dealing with the indices, to make sure that
the variable $r_n$ is the mark of the point $D_n$.

\begin{theorem}\label{burke1}
The marked point process $(\cD,r)$ has the same law as the marked
point process $(\cA,s)$.
\end{theorem}

 %% This result is an extension of the
%%  classical Burke Output Theorem \cite{burk}. The proof used here is
%%  based on the reversibility of $Q$ following \cite{reic}. See
%% \S \ref{sse-comm} for a complete discussion of the literature.

\begin{proof}
Knowing $Q$, one can recover the input of the model:
$(A_n,s_n)_n =\varphi(Q)$. We are going to prove that

$$\varphi\circ \cR (Q) = (-D_{-n+Q(0)+1},r_{-n+Q(0)+1})_n=\cR(\cD,r)$$

Let $\cT_n$ be the operator mapping a point process to its
$n$-th point, with the convention that $\cT_0(\cdot) \leq 0 <
\cT_1(\cdot )$. For instance, $A_n=\cT_n(\cA)=\cT_n\circ\cN_+(Q)$.
Taking into
account the shift in the indices and using Lemma \ref{numbering},
we get that $\cT_n (\cD)=\cT_{n}\circ\cN_-(Q)= D_{n-Q(0)}$. So we have

\begin{eqnarray*}
s_n & = &  D_n-\max(A_n,D_{n-1}) \\
& = &  \cT_{n+Q(0)}\circ\cN_-(Q)-\max(\cT_n\circ\cN_+(Q),\cT_{n+Q(0)-1}\circ\cN_-(Q))\:.
\end{eqnarray*}

Set $s_n=\cS_n(Q)$. Now let us apply the operators $\cT_n\circ\cN_+$
and $\cS_n$ to the reversed process $\cR(Q)$:

\begin{equation*}
\cT_n\circ\cN_+\circ\cR(Q) =  -D_{-n-Q(0)+1}
\end{equation*}
Furthermore $\cT_{n+\cR(Q)(0)-1}\circ\cN_-\circ\cR(Q) =
-A_{-n-Q(0)+2}$.
Thus,
\begin{eqnarray}
\cS_n\circ\cR(Q)&=&  -A_{-n-Q(0)+1}\: -\:\max(\:-D_{-n-Q(0)+1},\:
              -A_{-n-Q(0)+2}\:) \nonumber \\
              &=&   - A_{-n+Q(0)+1}\:+\: \min(\:D_{-n+Q(0)+1},\:
A_{-n+Q(0)+2}\:)\nonumber \\
              &=& r_{-n+Q(0)+1}\:. \label{cs-r}
\end{eqnarray}
Hence we obtain that $\varphi\circ\cR(Q)=\cR(\cD,r)$.

The process $Q$ is a stationary birth-and-death
process, hence reversible: $\cR(Q)\sim Q$. It implies that
$(\cA,s)=\varphi(Q)\sim \varphi\circ\cR(Q)=\cR(\cD,r)$.
Hence, $\cR(\cD,r)$ is an homogeneous Poisson process marked with an i.i.d.
 sequence of
r.v.'s. So, $\cR(\cD,r)\sim (\cD,r)$, which concludes the proof.
\end{proof}

\begin{corollary}
In the queueing model, the departure process $\cD$ is a Poisson
process of intensity
$\lambda$. In the storage model, the sequence $(r_n)_n$ of the
amounts of product $P$ departing at successive slots, is a
sequence of i.i.d. exponential r.v.'s of parameter $\mu$.
\end{corollary}

\subsection{Output theorem in the geometric case} \label{se-geom}

Let $\cA$ be a Bernoulli point process of parameter $p\in (0,1)$,
that is: all the points are integer valued, there is a point at a
given integer with probability $p$, and the presence of points at
different integers are independent. As before, set $\cA=(A_n)_n$
with $A_0\leq 0 < A_1$ and $a_n=A_{n+1}-A_n$. Then the sequence
$(a_n)_{n\geq 1}$ is a sequence of i.i.d. geometric r.v.'s with
parameter $p$ ($\forall k\in \N^*$, $P\{ a_1 = k \}
=(1-p)^{k-1}p$). Let $(s_n)_n$ be a sequence of i.i.d. geometric
r.v.'s with parameter $q\in (0,1)$, independent of $\cA$. We
assume that $p<q$ (stability condition).

\medskip

Let the marked point process $(\cA,s)$ be the input of the model
of Section \ref{se-model}. As in \S \ref{sse-exp}, the model is
stable and the output $(\cD,r)$ is a marked point process.
%% The
%% queue length process is a reversible process (Birth and Death
%% process) with invariant measure geometric distribution with
%% parameter $\frac{p(1-q)}{q(1-p)}$.

 We prove a Burke-type Theorem using an original approach based on
 a symmetrization of the workload
 process.

%% Define
%% \begin{equation}\label{eq-workload}
%% W(t)= \bigl[ w_n + s_n -(t-A_n) \bigr]^+, \ \ \text{ for } t \in
%% [A_n,A_{n+1}) \:,
%% \end{equation}
%% where $(w_n)_n$ is defined in \eref{Devolindley}. Observe that
%% $W(A_n^-)=w_n$. For the queueing model, $W(t)$ is the total amount
%% of service remaining to be done by the server at instant $t$, and
%% $W=(W(t))_t$ is called the {\em workload} process.
%% We now propose an alternative way of defining $W$ and we define at the
%% same time the dual process $\overline{W}$.

Define the (random) c\`adl\`ag maps
$f_n,g_n: \Omega\times \R \rightarrow \R$, by
\[
f_n (t) = \begin{cases} 0 & \mrm{if } t<A_n \\
                        D_n - t & \mrm{if } A_n \leq t \leq D_n
                        \\
0 & \mrm{if } t > D_n
\end{cases}, \qquad
g_n (t) = \begin{cases} 0 & \mrm{if } t<A_n \\
                        t-A_n & \mrm{if } A_n \leq t \leq D_n
                        \\
0 & \mrm{if } t > D_n
\end{cases}\:.
\]
Set
\begin{equation}\label{eq-workload}
W(t)=\bigvee_{n\in \Z} f_n(t), \qquad \overline{W}(t)=\bigvee_{n\in
  \Z} g_n(t) \:.
\end{equation}
Since $D_n=A_n+w_n+s_n$, we have $W(A_n^-)=w_n$. For the queueing model, $W(t)$ is the total amount
of service remaining to be done by the server at instant $t$, and
$W=(W(t))_t$ is called the {\em workload} process. The process $\overline{W}=(\overline{W}(t))_t$ is the dual of the process
  $W$, see Figure \ref{figure1} or Figure \ref{fi-zigzag} for an
illustration.

Looking at $W$ and $\overline{W}$ together is a way to symmetrize the workload
process. Consider the ordered pair of processes
$P=(W,\overline{W})$. Define
\[
\cR(P)=(\cR(\overline{W}),\cR(W))\]
(observe that the first element of the pair
is now $\cR(\overline{W})$, as opposed to $W$ for the pair $P$).
Obviously, we can recover the input of the model knowing $P$: $(\cA,s)
= \Psi(P)$. Also it is clear (see Figure \ref{figure1} or Figure \ref{fi-zigzag})
that: $\cR(\cD,r) = \Psi\circ\cR (P)$.

The crux of the argument in proving a Burke-type Theorem is now the next
Lemma.

\begin{lemma}\label{le-revers}
The ordered pair of the workload process and its dual is
reversible, that is: $(W,\overline{W})\sim (\cR(\overline{W}),\cR(W))$.
\end{lemma}

\begin{proof}
Clearly, $W$ and $\overline{W}$ are (time) stationary and ergodic processes.
An {\em idle period} is a maximal non-empty time interval
during which $W(t)=0$ (or equivalently, $\overline{W}(t)=0$). The {\em busy
  periods} are the time intervals
between the idle periods.
Hence $\R$ is partitioned by an alternating sequence of busy and idle
periods.

Let $(i_n)_{n\in \Z}$ and $(b_n)_{n\in \Z}$ be the lengths of
the successive idle periods and busy periods respectively. Clearly,
the sequences $(i_n)_n$ and $(b_n)_n$ are independent and
i.i.d. Moreover, it follows from the memoryless property of the
geometric distribution that $i_n$ is a geometric distribution of
parameter $p$.
So, to prove the reversibility of $(W,\overline{W})$, it is enough to
prove the
reversibility within a single busy period.

To do so, the easiest is to introduce an auxiliary process. Consider a
busy period which is assumed to consist, for simplicity, of the
$k$ customers numbered from 1 up to $k$.
%% For $n\in \{0,\dots,k-1\}$, set $v_n=D_n-A_n=w_n+s_n$ and $V_n
%% =\sum_{i=0}^n v_i$
%% Define
%% \begin{equation}
%% Z(t)= g_0(t) \vee f_0(t-D_n) \vee
Set for $n\in \{1,\dots k\}$:
\begin{equation}
C_n = B_n + s_n, \quad B_{n+1}=C_{n} + a_{n}\:.
\end{equation}
This defines $(B_n)_n$ and $(C_n)_n$ up to a translation. Define
\begin{equation}
\label{eq-zigzag}
Z(B_1)=0, \qquad Z(t) = \begin{cases} Z(B_n) + (t-B_n) & \text{for } t 
\in [B_n,C_n) \\
\bigl[Z(C_n) - (t-C_n) \bigr]^+ & \text{for } t \in [C_n,B_{n+1})
\end{cases}\:.
\end{equation}
On an interval of type $[B_n,C_n)$, we have $dZ/dt=1$ and on an
interval of type $[C_n,B_{n+1})$, we have $dZ/dt=(-1)\1_{Z>0}$.

Now, on a busy period, the process $Z$ is in
bijection with $(W,\overline{W})$, and the time-reversed process $\cR(Z)$ is in bijection with
$(\cR(\overline{W}),\cR(W))$.
This is illustrated in Figure \ref{fi-zigzag}, where we have
represented a trajectory of
$(W,\overline{W})$ and the corresponding trajectory of $Z$. The
process $Z$ is obtained by  gluing together the slices of
$\overline{W}$ and $W$ alternatively.

\begin{figure}[hbt]
\[\epsfxsize=400pt \epsfbox{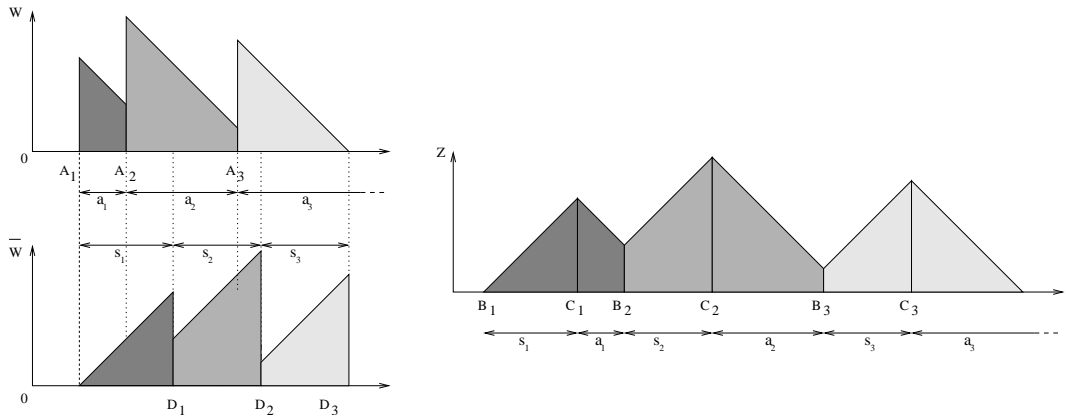} \]
\caption{The workload process, the dual workload process, and the zigzag process} \label{fi-zigzag}
\end{figure}

Hence, proving the reversibility of $(W,\overline{W})$ on a busy
period, is equivalent to proving the reversibility of $Z$.
We may call $Z$
the {\em zigzag} process. This process was considered (with a different
 motivation) in \cite{HMOc} and in particular, it was proved there
 that the zigzag process is reversible (although not Markovian). We recall 
the
 argument for completeness.

Identify $Z$ with the finite sequence of the lengths of its intervals of
increase and decrease. Let $(l_1,\dots ,l_{2k})\in (\N^*)^{2k}$,
with  $\sum_{i} l_{2i} = \sum_{i} l_{2i+1}=L$, be a possible value
for $Z$. (The total length of the periods of increase has to be equal to the total
length of the periods of decrease).
Recall that the variables $(s_n)_n$ and $(a_n)_n$ are
geometrically distributed of respective parameters $q$ and $p$. We
get:
\begin{eqnarray*}
P\bigl\{Z=(l_1,\dots ,l_{2k}) \bigr\} & = & (1-p)^{l_1-1} p (1-q)^{l_2-1}q \cdots
(1-p)^{l_{2k-1}-1}p (1-q)^{l_{2k}-1} \\
& = & (1-p)^{L-k} p^k
(1-q)^{L-k}q^{k-1}\:.
\end{eqnarray*}
Hence, given $L$ and $k$, all the different possible trajectories for
$Z$ are equiprobable. It implies that the same is true for the
time-reversed process $\cR(Z)$.
This completes the proof.
\end{proof}

We are now ready to prove the following result:

\begin{theorem}\label{burke-geom}
The marked point process $(\cD,r)$ has the same law as the marked
point process $(\cA,s)$.
\end{theorem}
%\begin{proof}
%The proof is based on the reversibility of $Q$. Let $\tilde{Q}$ be
%a birth and death process having the distribution as $Q$. Since it
%is not possible to reconstruct the arrivals and services from the
%queue-length process, we introduce additional information. This
%auxiliary reversible process indexed by $\frac{1}{2}+\Z$ contains
%the information on the events $\{\tilde{Q}(n-1)=\tilde{Q}(n)\neq
%0\}$. Consider the i.i.d. Bernoulli process
%$(M_{n+\frac{1}{2}})_{n\in\Z}$ such that
%
%$$P(M_{n+\frac{1}{2}}=1)=\frac{pq}{pq+(1-p)(1-q)}=1-P(M_{n+\frac{1}{2}}=0)\:.$$
%
%On the events $\{\tilde{Q}(n-1)=\tilde{Q}(n)\neq 0\}$, this
%process tells us whether there was {\em hidden} arrival and
%departure or nothing happened. This way we complete the lacking
%information and we follow the proof of the exponential case.
%
%\end{proof}
\begin{proof}
The proof is similar to the one of Theorem \ref{burke1} with $(W,\overline{W})$
playing the role of $Q$.

Recall that $(\cA,s) = \Psi (W,\overline{W})$ and $\cR(\cD,r) =
\Psi\circ\cR (W,\overline{W})$.
According to Lemma \ref{le-revers}, we have
\[
(\cA,s) = \Psi (W,\overline{W})
\sim \Psi\circ\cR (W,\overline{W}) = \cR(\cD,r)\:.
\]
We conclude that $(\cD,r)\sim \cR(\cA,s)\sim (\cA,s)$, where the last
equivalence comes from the fact that $(\cA,s)$ is a Bernouilli process
with i.i.d. marks, hence is reversible.
\end{proof}

%% \begin{remark}
%% \begin{itemize}
%% \item [$(i)$] The zigzag process $Z$ is not markovian.
%% \item [$(ii)$]  The zigzag process is clearly also reversible in
%% the exponential model. Hence the proof used for Theorem
%% \ref{burke-geom} can also be used to get Theorem \ref{burke1}.
%% \end{itemize}
%% \end{remark}

\subsection{Comments on the different proofs of Burke
Theorem}\label{sse-comm}

Reflecting on the above, there are three different ways to prove
Burke Theorem, be it in the exponential or geometric case. The
first way is by using analytic methods, the second is by using the
reversibility of the queue length process $Q$, and the third is by
using the reversibility of the zigzag process $Z$.

The original proof of Burke is for the exponential model using
analytic methods \cite{burk}. For the geometric model,
an analytic proof was given by Azizo\~{g}lu
and Bedekar \cite{BeAz}. For the exponential model, the idea of
using the reversibility of $Q$ to get the result is due to Reich
\cite{reic}. This proof has become a cornerstone of queueing
theory, it has been extended to various contexts and has given
birth to the concept of {\em product form networks}
\cite{brem,kell79}. Reich's proof does not translate directly to
the geometric model. Of course, $(Q(n))_{n\in \Z}$ is a reversible
birth-and-death Markov chain. However, a difficulty arises: It is
not possible to reconstruct $\cA$ and $s$ from $Q$. Indeed, on the
event $\{Q(n-1)= Q(n)>0\}$, two cases may occur: there is either
no departure and no arrival at instant $n$, or one departure and
one arrival; and it is not possible to distinguish between them
knowing only $Q$. One feasible solution is to add an auxiliary
sequence that contains the lacking information but the details
become quite intricate. This program has been carried out in
\cite{HsBu}, see also \cite[Theorem 4.1]{KOR}, for a variant: the geometric model with
{\em unused services}. The above idea of using the zigzag process to
prove Burke Theorem is original. Observe that the zigzag process is
also clearly reversible in
the exponential model. Hence the proof used for Theorem
\ref{burke-geom} could also be used to get Theorem \ref{burke1}.

In the original references \cite{burk,reic} and in all the
classical textbooks presenting the result \cite{brem,kell79,robe03},
the version proved is:
$\cA \sim \cD$: (``Poisson Input Poisson Output''). In
\cite{BeAz}, or \cite{HsBu} for the variant with {\em unused services}, 
the result proved is $s\sim r$. The complete
version $(\cA,s) \sim (\cD,r)$ appears first in \cite[Theorem
3]{OY02} for the exponential model with {\em unused
services}. This is extended to the geometric model with {\em
unused services} in \cite{KOR}. Brownian analogues are proved in
\cite{HaWi,OY02b}.

\subsection{Non-colliding random walks}

Following the lines of thought in \cite{OY02,KOR,HMOc}, it is possible to use
Theorems \ref{burke1} and \ref{burke-geom} to get representation
results for non-colliding random walks.

We start by stating a preliminary lemma.
\begin{lemma}
The Lindley's Equation \eref{Reclindley} admits a backward analogue:
\begin{equation}\label{rec-lindley}
w_n=(w_{n+1}+r_n-d_{n-1})^+\:.
\end{equation}
\end{lemma}
\begin{proof}
Recall Equation \eref{back2}: $r_n=s_n+w_n-w_{n+1}$. Using
\eref{eq-rel0}, we get
\begin{equation}\label{eq-etoile}
d_n = a_n + s_{n+1} -r_n = a_{n}-w_n+w_{n+1}-s_n+s_{n+1} \:.
\end{equation}
Therefore,
\begin{eqnarray*}
w_{n+1}+r_n-d_{n-1}&=&w_{n+1}+s_{n}+w_n-w_{n+1}-a_{n-1}-w_n+w_{n-1}-s_n+s_{n-1} \\
&=&w_{n-1}+s_{n-1}-a_{n-1} \:.
\end{eqnarray*}
Taking the maximum with $0$ on both sides of the equality and
using Lindley's equation \eref{Reclindley}, we obtain
$(w_{n+1}+r_n-d_{n-1})^+=w_n$.
\end{proof}
\begin{proposition}\label{pr-pitman}
Let the sequences of r.v.'s $(a_n)_{n\in \N^*}$ and $(s_n)_{n\in \N^*}$ be
as in \S \ref{sse-exp} or as in \S \ref{se-geom}.  The conditional law of
$(\sum_{i=1}^{n} a_i,\:\sum_{i=1}^{n}s_i)$, given that
$\{\sum_{i=1}^{k}a_i \geq \:\sum_{i=1}^{k+1}s_i,\ \forall k \geq
0\}$ is the same as the unconditional law of $(\max_{1\leq j \leq
n}\{\sum_{i=1}^{j}a_i+\sum_{i=j+1}^{n+1}s_i\},\ \min_{1\leq j \leq
n}\{\sum_{i=2}^{j}s_i+\sum_{i=j+1}^{n}a_i\})$.
\end{proposition}

\begin{proof}
Set $v_n=w_n+s_n=w_{n+1}+r_n$ (the {\em sojourn time} of customer
$n$ in the queue). By developing \eref{rec-lindley}, we obtain
 \begin{equation}\label{back-lind}
v_n=w_{n+1}+r_{n}=\sup_{k\geq n}\{\sum_{i=n}^{k}r_i -
\sum_{i=n}^{k-1}d_i\}\:.
 \end{equation}

By summing the equalities \eref{eq-etoile}, we obtain on the event
$\{v_1=0\}$:
\begin{eqnarray*}
 \sum_{i=1}^{n}d_i&=& \sum_{i=1}^{n}a_i+w_{n+1}+s_{n+1} \\
 &=&\sum_{i=1}^{n}a_i+\max_{1\leq j\leq n}\{\sum_{i=j+1}^{n+1}s_i -
\sum_{i=j+1}^{n}a_i\} \
 = \ \max_{1\leq j\leq n}\{\sum_{i=1}^{j}a_i +
\sum_{i=j+1}^{n+1}s_i\}\:,
\end{eqnarray*}
and summing \eref{back2} on the event $\{v_1=0\}$, we get:
\begin{eqnarray*}
 \sum_{i=1}^{n}r_i&=& \sum_{i=2}^{n}s_i-w_{n+1}\\
 &=&\sum_{i=2}^{n}s_i-\max_{1\leq j\leq n}\{\sum_{i=j+1}^{n}s_i -
\sum_{i=j+1}^{n}a_i\} \
 = \ \min_{1\leq j\leq n}\{\sum_{i=2}^{j}s_i +
\sum_{i=j+1}^{n}a_i\}\:.
\end{eqnarray*}
To summarize, we have
\begin{equation}\label{Representation}
 \sum_{i=1}^{n}d_i=\max_{1\leq j \leq
n}\{\sum_{i=1}^{j}a_i+\sum_{i=j+1}^{n+1}s_i\},\:\:\:
 \sum_{i=1}^{n}r_i=\min_{1\leq j \leq
 n}\{\sum_{i=2}^{j}s_i+\sum_{i=j+1}^{n}a_i\}\:.
\end{equation}

Applying Theorem \ref{burke1} or \ref{burke-geom}, the
distribution of $(\sum_{i=1}^{n}a_i,\:\sum_{i=1}^{n}s_i)$ given
 that  $\{\sum_{i=1}^{k}a_i \geq \:\sum_{i=1}^{k+1}s_i,\ \forall k \geq
0\}$ is the same as the law of
 $(\sum_{i=1}^{n}d_i,\:\sum_{i=1}^{n}r_i)$ given
 that $\{\sum_{i=1}^{k}d_i \geq
\:\sum_{i=1}^{k+1}r_i,\ \forall k \geq
 0\}$.
 By  (\ref{back-lind}), the event $\{\sum_{i=1}^{k}d_i \geq
 \:\sum_{i=1}^{k+1}r_i,\ \forall k \geq
 0\}$ is equal to $\{v_1=0\}$.

 Using \eref{Representation}, we have that the distribution of $(\sum_{i=1}^{n}d_i,\:\sum_{i=1}^{n}r_i)$ given
 that  $\{v_1=0\}$
% $\{\sum_{i=1}^{k}a_i \geq \:\sum_{i=1}^{k+1}s_i,\ \forall k \geq 0\}$
is the same as the law of $(\max_{1\leq j \leq
n}\{\sum_{i=1}^{j}a_i+\sum_{i=j+1}^{n+1}s_i\},\min_{1\leq j \leq
 n}\{\sum_{i=2}^{j}s_i+\sum_{i=j+1}^{n}a_i\})$ given
 that $\{v_1=0\}$. This last conditional law is the same as the
 unconditional law of  $(\max_{1\leq j \leq
n}\{\sum_{i=1}^{j}a_i+\sum_{i=j+1}^{n+1}s_i\},\min_{1\leq j \leq
 n}\{\sum_{i=2}^{j}s_i+\sum_{i=j+1}^{n}a_i\})$. Indeed,
the variables $(a_i)_{i\geq 1}$ and $(s_i)_{i\geq 2}$ are
independent of $\{v_1=0\}$.

By fitting all the parts together, we obtain the desired result.
\end{proof}

It is possible to extend Proposition \ref{pr-pitman} to the limiting case
$E[a_1]=E[s_1]$, and also to higher dimensions, by adapting
the methods of \cite{OY02,KOR,HMOc} to the present setting.

\section{Transient Behavior and RSK Representation}\label{se-trans}

In Sections \ref{sse-satu} and \ref{sse-rsk}, the model described and
the results obtained have a purely combinatorial nature. In
particular, the results are true pathwise, without any 
probabilistic assumption. 
These results are then used in Section \ref{sse-inter}
for geometric queues or stores: here the
pathwise arguments are coupled with probabilistic ones. 

\subsection{The saturated tandem}\label{sse-satu}

We consider another aspect of the dynamic of queues and stores: the
transient evolution for the model starting empty. More precisely,
consider the model of \S \ref{se-model} under the assumption that
$w_0=A_0=s_0=0$ (which implies $D_0=r_0=0$) and focus on the customers,
resp. time slots, from 1 onwards.

It is convenient to describe such a model with a different
perspective. We first do it for the queue. View the arrivals as being
the departures from a {\em virtual queue} having at instant 0 an infinite
number of customers (labelled by $\N^*$) in its buffer. The service
time of customer $n$ in the virtual queue is $a_{n-1}$. We describe
this as a {\em saturated tandem of two queues}.

Let us turn our attention to the store. View the supplies as being the
departures from a virtual store having an infinite stock at the end of
time slot 0. In the virtual store, the request (=departure) at time slot
$n$ is
$s_{n}$. This
is a {\em saturated tandem of two stores}.

Now we want to fit these two descriptions together.
Denote the virtual queue/store as {\em queue/store 1} and the other one as
{\em
queue/store 2}.
For convenience, set $u(n,1)=a_{n-1}$ and $u(n,2)=s_{n}$ for all $n\geq
 1$.
The saturated tandem is completely specified
by the family
$(u(n,i), n \in \N^*,i=1,2)$ of input variables.
These variables are the services, resp. requests, when
the model is seen as a tandem of queues, resp. stores.
Next table gives the input variables in the saturated tandems of two
queues/stores.

\medskip

\begin{center}

\begin{tabular}{|c||c|c|c|c|}\hline
Customer / Time slot & 1 & 2 & $\cdots$ & $n$ \\
\hline \hline Queue 2 / (Virtual) Store 1& $u(1,2)=s_1$ &
$u(2,2)=s_2$ & & $u(n,2)=s_n$ \\ \hline (Virtual) Queue 1 / Store
2  & $u(1,1)=a_0$ & $u(2,1)=a_1$ & &$u(n,1)=a_{n-1}$
\\ \hline
\end{tabular}
\end{center}

\medskip

A couple of observations are in order.
Observe that $u(n,i)$ is the service of customer $n$  in queue $i$, and
the request at time slot $n$ in store $3-i$. In other words,
the elements (queues or stores) associated with a given sequence
$(u(\cdot,i)), i=1,2,$ are crossed in reverse orders in the
queueing/storage tandem. This is illustrated in Figure \ref{fi-tandem} (set
$K=2$).
Observe also that there is a shift in the time slots for the storage
model: the departure from store 1 at time slot $n$ is the supply of
store 2 at time slot $n+1$ (contrast this with the situation for the
queues).  This is consistent with
a model in which we view a time slot as being decomposable in three
consecutive stages: first, the supply arrives; second, the request is
made; third, the departure occurs. See Figure \ref{fi-store}.

\begin{figure}[hbt]
\[\epsfxsize=280pt \epsfbox{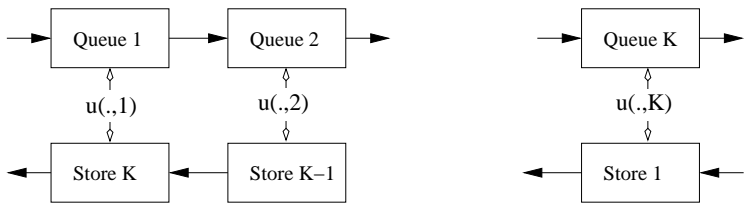} \]
\caption{Queues and stores are gone through in opposite orders}
\label{fi-tandem}
\end{figure}

The above setting is naturally extended to define the {\em saturated
  tandem of $K$ queues/stores}. Such a model is entirely defined by a
  family of $\N$-valued variables $(u(i,j))_{i\in \N^*, j\in
  \{1,\dots, K\}}$. (In \S \ref{se-model}, the variables were 
  valued in $\N^*$. This restriction is not necessary here.) 
For the queueing model: (i) at instant 0, queue 1 has an infinite number of
  customers labelled by $\N^*$ in its buffer, and the other queues are
  empty; (ii)
  $u(i,j)$ is the service of customer $i$ at
  queue $j$; (iii) the instant of departure of customer $n$ from queue
$i$  is the instant of arrival of customer $n$ in
queue $i+1$.
For the storage model: (i) at the end of time slot 0, store 1 has an
  infinite stock and the other stores have an empty stock; (ii)
  $u(i,K+1-j)$ is the request at time slot $i$ in store $j$; (iii) the
  departure at time slot $n$ from store
$i$ is the supply at time slot $n+1$ in store $i+1$. The models
are depicted in Figure \ref{fi-tandem}.

\subsection{Interacting particles representation}
\label{sse-part}

It is helpful to describe the above tandem models as interacting
particle models (see for instance \cite{ligg99}). 
Commonly, interacting particle systems are defined as
continuous-time Markov processes. 
In contrast, the
description below is pathwise and the time is discrete. 

\medskip

Classically, the queues in tandem may be viewed as the following {\em
  zero-range process}. The set of sites is $\{1,\dots , K\}$ and a
  site is occupied by a number of particles in $\N \cup
  \{+\infty\}$. At each site, 
  the front particle (if any) has a clock whose value is decremented
  by 1 at each time slot. When it reaches 0, the particle jumps to the site
  on its right. The right-most site has an infinite number of
  particles. This is illustrated in Figure \ref{fi-tandq}. The
  variables $(u(i,j))_{i}$ correspond to the  
initial values of the clocks of the successive front particles at site
  $j$. 

\begin{figure}[hbt]
\[\epsfxsize=260pt \epsfbox{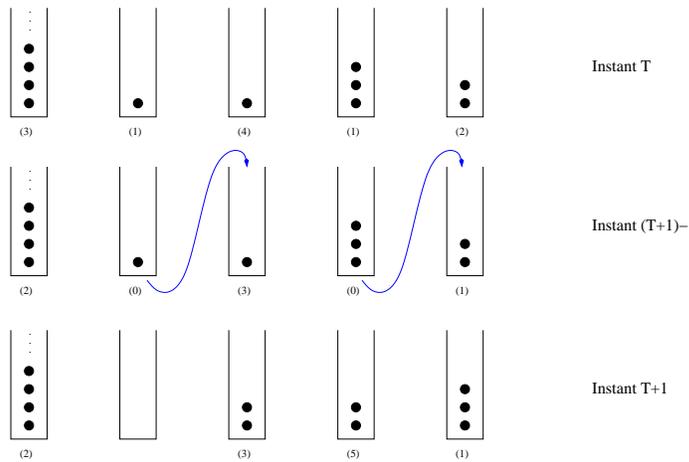} \]
\caption{The zero-range process for queues in tandem}
\label{fi-tandq}
\end{figure}

\medskip

We now describe the stores in tandem as an original
zero-range model: the {\em bus-stop model}. 
At the beginning of each time slot, a bus arrives at each site. The
size of the bus varies from site to site, and from time slot to time
slot. 
The particles take place in the bus until it is filled in, and are
transported to the next site on the left, during the current time
slot. The left-most site has an infinite number of particles. 
This mechanism is illustrated in Figure \ref{fi-tands}. Here the
variables $(u(i,j))_{i}$ correspond to the successive bus sizes at
site $j$. 

\begin{figure}[hbt]
\[\epsfxsize=280pt \epsfbox{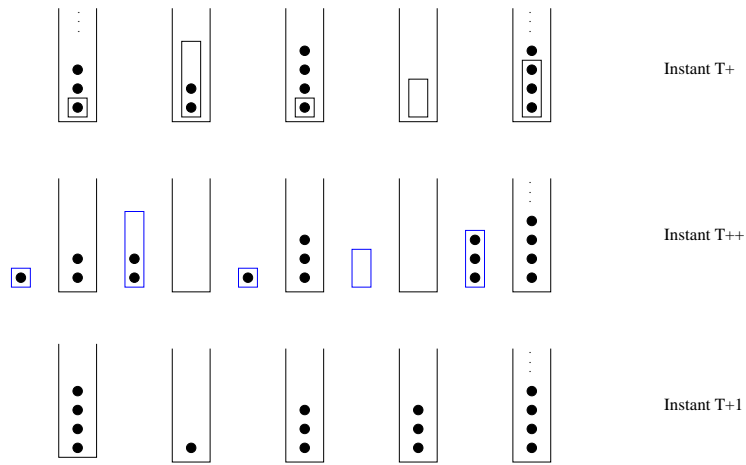} \]
\caption{The zero-range process for stores in tandem}
\label{fi-tands}
\end{figure}

\medskip

Alternatively, the above two models can also be viewed as {\em exclusion
processes}. The set of sites is embedded in $\N$ and each site
now contains either 0 or 1 particle. 

The queues in tandem are associated with the classical TASEP or
totally asymmetric exclusion process. 
For the stores in tandem, the exclusion process is as
follows. Particles take jumps of integer length to the right which are
constrained by the interdiction to overpass other particles. In one
time step, particles jump in order from left to right. 
This model is close but different from the exclusion processes
considered in \cite{RSSS,sepp00}. 

\begin{figure}[hbt]
\[\epsfxsize=300pt \epsfbox{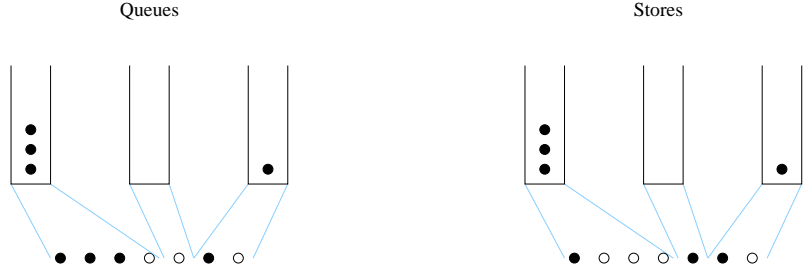} \]
\caption{Exclusion processes for the queues and the stores}
\label{fi-exclu}
\end{figure}

The mapping between zero-range process and exclusion process, both for
the queues and the stores, is illustrated in Figure \ref{fi-exclu}.

\subsection{Robinson-Schensted-Knuth representation}\label{sse-rsk}

A {\em partition} of $n\in \N^*$ is a sequence of integers
$\lambda=(\lambda_1,\dots,\lambda_k)$ such that $\lambda_1\geq
\cdots\geq\lambda_k\geq 0$ and $\lambda_1 +
\cdots +\lambda_k =n$. We use the notation $\lambda \vdash n$. By
convention, we identify partitions having the same non-zero
components.
%% Consider a sequence $(\lambda_1,\dots,\lambda_K)$ of integers such
%% that $\lambda_1\geq \dots\geq\lambda_K>0$ and set
%% $\lambda_1+\dots+\lambda_K=M$.
The {\em (Ferrers) diagram} of
$(\lambda_1,\dots,\lambda_k)\vdash n$ is a collection of $n$ boxes
arranged in left-justified rows, the $i$-th row starting from the
top consisting of $\lambda_i$ boxes.
A {\em (semi-standard Young) tableau} on the alphabet $\{1,\dots,\ell\}$ is
a diagram in which each box is filled
in by a label from $\{1,\dots,\ell\}$ in such a way that the entries are
weakly increasing from left to right along the rows and strictly
increasing down the columns. The {\em shape} of a diagram or tableau
is the underlying partition. A {\em standard tableau} of size
$n$ is a tableau of shape $\lambda\vdash n$ whose entries are from $\{1,\dots , n\}$ and are
distinct.
In Figure \ref{fi-Young} we have represented a diagram on the left and
a tableau of the same shape on the right.

\begin{figure}[hbt]
\[\epsfxsize=170pt \epsfbox{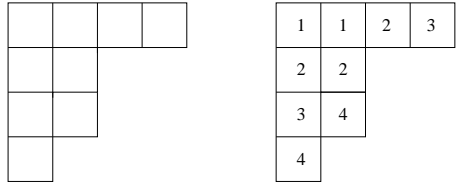} \]
\caption{Ferrers diagram and semi-standard Young tableau of shape
  $(4,2,2,1)\vdash 9$}
\label{fi-Young}
\end{figure}

The Robinson-Schensted-Knuth row-insertion algorithm (RSK algorithm) takes a tableau $T$ and $i\in \N^*$ and
constructs a new tableau $T\leftarrow i$. The tableau $T\leftarrow i$ has one more box than $T$
and is constructed as follows.
If $i$ is at least as large as the labels of the first (upper) row of $T$, add a box labelled $i$ to the end
of the first row of $T$ and stop the procedure. Otherwise, find the
leftmost entry in the first row which is strictly larger
than $i$, relabel the corresponding box by $i$ and apply the same procedure
 recursively
to the second row and to the bumped label. By convention, an empty row
has label $0$. With this convention, the above procedure stops.

Consider a word $v=v_1\cdots v_n$ over the alphabet
$\{1,\dots , k\}$. The tableau associated with $v$ is by definition
\[
P=(\cdots ((T_0\leftarrow  v_1)\leftarrow  v_2 )\cdots \leftarrow
v_n)\:,
\]
where $T_0$ is the empty tableau. Observe that $P$ has at most $k$
non-empty rows. Classically, the length of the top (and longest) row of $P$
 is
equal to the longest weakly-increasing subsequence in $v$.

\begin{remark}\rm While building the tableau $P$, it is possible to build another
tableau of the same shape in which the entries, labelled from 1 to
$|v|$, record the order in which the boxes are added. This {\em
  recording tableau} is a standard tableau of size
$|v|$. By doing this, one  defines a bijection between words
of $\{1,\dots ,k\}^n$ and ordered pairs of tableaux of the same shape, the
first being semi-standard over the alphabet $\{1,\dots,k\}$ and the
second being standard and of size $n$. This bijection is often
referred to as the Robinson-Schensted-Knuth correspondence. 
Here, we do not need this result and we do not consider the recording
tableau.
\end{remark}

Consider a family
$U=(u(i,j))_{(i,j)\in\{1,\dots,N\}\times\{1,\dots,K\}}$ of variables valued 
in $\N$.
(Here we do not make any
assumption on the variables, which may or may not be random.)
We associate
with $U$ the word $w(U)$ over the alphabet $\{1,\dots , K\}$ defined by
\begin{equation}\label{eq-wu}
\begin {array}{ccccc}
\noalign{\medskip}
w(U)=w_1\dots w_N, \mrm{ with } w_i=&\underbrace{1 \cdots 1}
&\underbrace{2\cdots 2}&\cdots&\underbrace{K \cdots K}\\
    &      u(i,1)&                 u(i,2)&          \cdots      &u(i,K)
\end {array}\:.
\end{equation}

Set $M=|w(U)|=\sum_{i=1}^N\sum_{j=1}^K u(i,j)$.
For $i=1,\dots,K$, define
\[
\forall n \leq M, \ \ x_i(n)=|\{j\leq n\:|\:w(U)_j=i\}|\:.
%% , \ \ \forall m\leq n, x_i(m,n)=x_i(n)-x_i(m)\:.
\]
Given two maps $x,y: \{1,\dots, N\} \longrightarrow \N$, define the
maps $x\triangledown y, x\triangleup y :
\{1,\dots, N\} \longrightarrow \N$ as follows: $\forall n\leq N$,
\begin{equation}\label{min-max}
\ x\triangledown y(n)=\max_{0\leq m
  \leq n}\Bigl[x(m)+y(n)-y(m)\Bigr], \ \
x\triangleup y(n)=\min_{0\leq m \leq n}\Bigl[x(m)+y(n)-y(m)\Bigr]\:.
\end{equation}

Denote by $P(U)$ the tableau
obtained from $w(U)$ by applying the RSK algorithm and let
$(\lambda_1,\dots , \lambda_K)$ be its shape. The following
holds
\begin{equation}\label{lambda}
\lambda_1 = x_1\triangledown x_2\triangledown\dots\triangledown
x_K(M), \ \ \lambda_K = x_K\triangleup\dots\triangleup x_{2}\triangleup
x_1(M)\:,
\end{equation}
where
the operations $\triangledown,\:\triangleup$  are performed from left
to right (the operations $\triangledown,\:\triangleup$ are
non-associative).
The expression for $\lambda_1$ follows from the fact that $\lambda_1$
is the longest weakly-increasing subsequence in $w(U)$. The expression for
$\lambda_K$ is proved in \cite[Theorem 3.1]{ocon03}.
In fact, the result from \cite{ocon03} is more general:
there exists a min-max-type operator $\Gamma_K$ such that
$(\lambda_1, \dots,\lambda_K)=\Gamma_K(x_1,\dots, x_K)$.

\medskip

A {\em lattice path} is a sequence $\pi=((i_1,j_1),\dots ,(i_l,j_l))$
with $(i_k,j_k)\in \Z^2$. The {\em steps} of $\pi$ are the differences
$(i_{k+1}-i_k,j_{k+1}-j_k),\:k=1,\dots l-1$.

\begin{figure}[hbt]
\[\epsfxsize=260pt \epsfbox{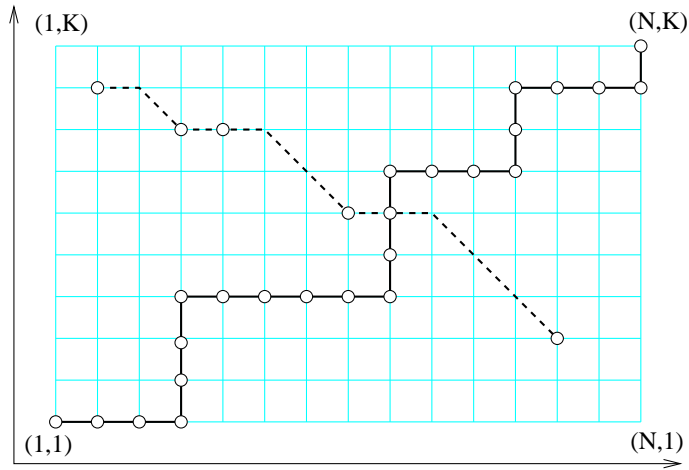} \]
\caption{A path from $\Pi$ (plain line) and a path from
  $\widetilde{\Pi}$ (dashed line)}
\label{fi-paths}
\end{figure}

Let $\Pi$ be the set of lattice paths from (1,1) to $(N,K)$ with
steps of the type $(1,0)$ or $(0,1)$. All the paths in $\Pi$ have
the same number of nodes: $N+K-1$. Let $\widetilde{\Pi}$ be the
set of lattice paths from one of the points in $\{(1+i,K-i),
i=0,\dots, K-1\}$ to one of the points in $\{(N-j,1+j), j=0,\dots
, N-1\}$ with steps of the type $(1+i,-i),\:i \geq 0$.

 Observe
that all the paths in $\widetilde{\Pi}$ also have the same number
of nodes: $(N-K+1)^+$. In particular, $\widetilde{\Pi}$ is empty
when $N<K$. An example of a path from each of the two sets is
provided in Figure \ref{fi-paths}.

The following result holds. When interpreting it, recall that
queues and stores are arranged in opposite orders in the tandems,
see Figure \ref{fi-tandem}. It is also fruitful to compare the
statements of Proposition \ref{pr-pitman} and Theorem
\ref{row-queue} for $K=2$.

\begin{theorem}\label{row-queue}
Consider a saturated tandem of $K$ queues/stores with variables
$u(i,j),i\in \N^*,j\in \{1,\dots, K\},$ as defined in \S
\ref{sse-satu}. Fix $N\in \N^*$. Let $D$ be the instant of
departure of customer $N$ from queue $K$ and let $R$ be the
cumulative departures over the time slots 1 to $N$ in store $K$.
Set $U=(u(i,j))_{1\leq i \leq N\:,\:1\leq j \leq K}$ and define
$w(U)$ as in \eref{eq-wu}. Let $(\lambda_1,\dots , \lambda_K)$ be
the shape of the tableau associated with $w(U)$. We have
\begin{eqnarray}
\lambda_1 &=& \max_{\pi\in\Pi}\Bigl[\sum_{(i,j)\in\pi}u(i,j)\Bigr] \ =
\ D \label{main1}\\
\lambda_K &=&
\min_{\pi\in\widetilde{\Pi}}\Bigl[\sum_{(i,j)\in\pi}u(i,j)\Bigr]\ = \
R \:. \label{main2}
\end{eqnarray}
\end{theorem}
First, we state and prove a preliminary result.
\begin{lemma}\label{le-lambbdaPath}
We have
\begin{equation}
x_K\triangleup\dots\triangleup x_{2}\triangleup
x_1(M)=\min_{\pi\in\widetilde{\Pi}}\Bigl[\sum_{(i,j)\in\pi}u(i,j)\Bigr]\:.
\end{equation}
\end{lemma}
\begin{proof}
By definition,
$$
x_K\triangleup\dots \triangleup
x_1(M)=\min_{1\leq m_K<\cdots<m_1=
M} x_K(m_K)+x_{K-1}(m_{K-1})-x_{K-1}(m_K)+\cdots+x_{1}(m_{1})-x_{1}(m_2) \:.
$$
Let $(m_K^*,\dots ,m_2^*)$ be a sequence of integers realizing the
minimum in the above equation.
Consider the set of integers
\[
I=\{ \sum_{i=1}^k |w_i| +
\sum_{j=1}^{K-1} u(k+1,j), \ k\in \{0,\dots, K-1\}\:.
\]
We may assume wlog that $m_K^*\in I$. Indeed assume that $m_K^*\not\in
I$, and let $n_K$ be the smallest integer such that $n_K \in I, n_K>
m_K^*$. Then, we clearly have $x_K(n_K)=x_K(m_K^*)$ and
\begin{eqnarray*}
 x_K\triangleup\dots \triangleup
x_1(M) - x_K(n_K) & \leq & \min_{n_K < m_{K-1} < \cdots < m_2}
x_{K-1}(m_{K-1})-x_{K-1}(n_K) + \cdots + x_1(M)-x_1(m_2) \\
&\leq & \min_{m_K^* < m_{K-1} < \cdots< m_2}
x_{K-1}(m_{K-1})-x_{K-1}(m_K^*) + \cdots + x_1(M)-x_1(m_2) \\
& = &  x_K\triangleup\dots \triangleup x_1(M) -x_K(m_K^*)\:.
\end{eqnarray*}
So we assume that  $m_K^*\in I$. Let $i_K$ be such that $m_K^*=
\sum_{i=1}^{i_K-1} |w_i| +
\sum_{j=1}^{K-1} u(i_K,j)$. We have  $x_K(m_K^*)
=\sum_{j=1}^{i_K-1}u(j,K)$. Then we prove with the same type of argument that
we can assume wlog that $m_{K-1}^*\in \bigl\{ \sum_{i=1}^k |w_i| +
\sum_{j=1}^{K-2} u(k+1,j), \ k\in \{i_K,\dots, K-1\}\bigr\}$.

\begin{figure}[hbt]
\[\epsfxsize=300pt \epsfbox{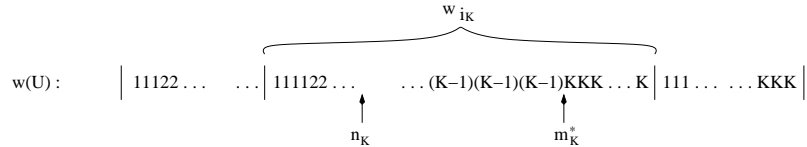} \]
\caption{Illustration of the proof of Lemma \ref{le-lambbdaPath}} \label{fi-proof}
\end{figure}

Since $m_K^*\in I$, we have $x_K(m_K^*) =\sum_{j=1}^{i_K}u(K,j)$ for
some $i_K$. Then we prove with the same type of argument that
we can assume wlog that $m_{K-1}^*\in \{ \sum_{i=1}^k |w_i| +
\sum_{j=1}^{K-2} u(k,j), \ k\in \{i_K+1,\dots, K-1\}$.
It implies that $x_{K-1}(m_{K-1})-x_{K-1}(m_K)=\sum_{j=i_K+1}^{i_{K-1}-1}u(K-1,j)$ for
some $i_{K-1}$. By repeating the argument, we obtain the desired
result.
%% We first minimize on the number of $K$, that is finding the
%% optimal $m_K$ for $x_K(m_K)$. It is not difficult to see that the
%% best way to do this is to take
%% $m_K=|w_1|+\dots+|w_{i_K-1}|+\sum_{i=1}^{K-1} u(i,i_K)$ for some $i_K$
%% (see Figure \ref{fi-optim}). Hence, $x_K(m_K) =\sum_{j=1}^{i_K-1}u(K=
%% (,j)$. 
%% %% \begin{figure}[hbt]
%% %% \[\epsfxsize=350pt \epsfbox{word-optim.eps} \]
%% %% \caption{Proof of lemma \ref{le-lambbdaPath}} \label{fi-optim}
%% %% \end{figure}
%% Then we continue our optimization over the factor
%% $w_{i_K+1}\cdots w_{N}$ that is
%% $x_{K-1}(m_{K-1})-x_{K-1}(m_K)=\sum_{j=i_K+1}^{i_{K-1}-1}u(K-1,j)$ f=
%% (or
%% some $i_{K-1}$, and so on.
\end{proof}

\renewcommand\proofname{Proof  of Theorem \ref{row-queue}}
\begin{proof}
The result for $\lambda_1$ is essentially due to Schensted. The
left equality in (\ref{main1}) is (\ref{lambda}) and the right
equality appears for instance in \cite{muth,TeWo,GlWh}. As far as
$\lambda_K$ is concerned, Lemma \ref{le-lambbdaPath} leads to
\[\lambda_K=x_K\triangleup\dots\triangleup x_{2}\triangleup
x_1(M)=\min_{\pi\in\widetilde{\Pi}}\Bigl[\sum_{(i,j)\in\pi}u(i,j)\Bigr]\:.\]

It remains to prove the right equality in \eref{main2}. Let us
denote by $r(i-1,j)$ the amount departing at time slot $i$ from
store $j$ and $w(i,j)$ the quantity remaining at the end of time
slot $i$ in store $j$. According to \eref{back2}, we have
\begin{equation}\label{eq-r_tandem}
r(n,k)=r(n-1,k-1)+w(n,k)-w(n+1,k)\:,
\end{equation}

Applying \eref{Reclindley} to our tandem model,
\begin{eqnarray*}
w(n+1,k)&=& \bigl[ w(n,k)+r(n-1,k-1)-u(n,K+1-k)\bigr]^+\\
&=&\max_{1\leq m\leq
n}[\sum_{i=m}^{n-1}(r(i,k-1)-\sum_{i=m+1}^{n}u(i,K+1-k))]\:.
\end{eqnarray*}
Using \eref{eq-r_tandem}, for $n\geq K$, we get

\begin{eqnarray*}
r(n,K)&=&r(n-1,K-1)+w(n,K)-w(n+1,K)\\
&=&r(n-2,K-2)+w(n-1,K-1)-w(n,K-1)+w(n,K)-w(n+1,K)\\
&=&r(n-K+1,1)+\sum_{i=0}^{K-2}\bigl[ w(n-i,K-i)-w(n-i+1,K-i)\bigr] \\
&=&u(n-K+1,K)+\sum_{i=0}^{K-1} \bigl[ w(n-i,K-i)-w(n-i+1,K-i)\bigr] \:.
\end{eqnarray*}
Hence,
\begin{eqnarray*}
R=\sum_{n=1}^{N}r(n,K)&=&\sum_{n=K}^{N}r(n,K)\\
&=&\sum_{n=K}^{N}u(n-K+1,K)+\sum_{n=K}^{N}\sum_{i=0}^{K-1}w(n-i,K-i)-w(n-i+1,K-i)\\
&=&\sum_{n=1}^{N-K+1}u(n,K)+\sum_{i=0}^{K-1}\sum_{n=K}^{N}w(n-i,K-i
)-w(n-i+1,K-i)\\
&=&\sum_{n=1}^{N-K+1}u(n,K)+\sum_{i=0}^{K-1}w(K-i,K-i)-w(N-i+1,K-i)\:.
\end{eqnarray*}
Now recall that $w(k,k)=0$, for all $k\geq 1$. We obtain
\begin{eqnarray}
\nonumber
R&=&\sum_{n=1}^{N-K+1}u(n,K)+\sum_{i=1}^{K}w(i,i)-\sum_{i=1}^{K}w(N-
K+i+1,i)\\\label{eq-R}
&=&\sum_{n=1}^{N-K+1}u(n,K)-\sum_{i=1}^{K}w(N-K+i+1,i)\:.
\end{eqnarray}
Moreover, set
$A=w(N+1,K)+w(N,K-1)$, we have

\begin{eqnarray*} 
A&=&\max_{1\leq i_1\leq
N}\Bigl[\sum_{n=i_1}^{N-1}r(n,K-1)-\sum_{n=i_1+1}^{N}u(n,1)\Bigr]+w(N,K-1)\\
&=&\max_{1\leq i_1\leq
 N}\Bigl[\sum_{n=i_1}^{N-1}[ r(n-1,K-2)+w(n,K-1)-w(n+1,K-1)]-
\sum_{n=i_1+1}^{N}u(n,1)\Bigr]+w(N,K-1)\\
&=&\max_{1\leq i_1\leq
 N}\Bigl[\sum_{n=i_1-1}^{N-2}r(n,K-2)+w(i_1,K-1)-w(N,K-1)-
\sum_{n=i_1+1}^{N}u(n,1)\Bigr]+w(N,K-1)\\
&=&\max_{1\leq i_1\leq
N}\Bigl[\sum_{n=i_1-1}^{N-2}r(n,K-2)+w(i_1,K-1)-\sum_{n=i_1+1}^{N}u(n,1)\Bigr]\\
&=&\max_{1\leq i_1\leq
  N}\Bigl[\sum_{n=i_1-1}^{N-2}r(n,K-2)+\max_{1\leq i_2\leq
  i_1-1}[\sum_{n=i_2}^{i_1-2}r(n,K-2)-\sum_{n=i_2+1}^{i_1-1}u(n,2)\Bigr]- \sum_{n=i_1+1}^{N}u(n,1)]\\
&=&\max_{1\leq i_2<i_1\leq
  N}\Bigl[\sum_{n=i_2}^{N-2}r(n,K-2)-\sum_{n=i_2+1}^{i_1-1}u(n,2)-\sum_{n=
i_1+1}^{N}u(n,1)\Bigr]\:.
\end{eqnarray*}

Applying this repeatedly leads to
\begin{equation}\label{eq-sum_w}
\sum_{i=1}^{K}w(N-K+i+1,i)=\max_{1\leq i_{K-1}<\dots<i_1\leq
N}\Bigl[\sum_{n=i_{K-1}}^{N-K+1}u(n,K)-\sum_{n=i_{K-1}+1}^{i_{K-2}-1}u(n,K-1)+\dots
+\sum_{n=i_1+1}^{N}u(n,1)\Bigr]\:.
\end{equation}
Combining \eref{eq-R} and \eref{eq-sum_w}, we obtain
\begin{eqnarray*}
R&=&\min_{1\leq i_{K-1}<\dots<i_1\leq
N}\Bigl[\sum_{n=1}^{i_{K-1}-1}u(n,K)+\sum_{n=i_{K-1}+1}^{i_{K-2}-1}u(n,K-1)+\dots+
\sum_{n=i_1+1}^{N}u(n,1)\Bigr] \\
&=&\min_{\pi\in\widetilde{\Pi}}\Bigl[\sum_{(i,j)\in\pi}u(i,j)\Bigr]\:.
\end{eqnarray*}

\end{proof}

Let us comment on Theorem \ref{row-queue}. 
First, it is interesting to observe
that the 6 quantities appearing in the equalities \eref{main1} and
\eref{main2} all come from the same model. Now, 
consider the two identities \eref{main1} and \eref{main2}
separately. The first one is classical and has led to several
interesting developments for queues in series. A by-product
of the newly proved identity \eref{main2} is therefore to pave the way to obtain
analog results for stores in tandem. We set out this program in
some details in the next Section. 

\subsection{Stochastic models, interchangeability, and hydrodynamic
  limits} 
\label{sse-inter}

If $T$ is a tableau over the alphabet $\{1,\dots, k\}$, we write
$x^T=\prod_{i=1}^{k} x_i^{\alpha_i}$, where $\alpha_i$ is the
number of occurrences of the label $i$ in the tableau. The {\em
Schur function} $s_{\lambda}$ associated with the partition
$\lambda=(\lambda_1,\dots, \lambda_k)$ is defined by
\[
s_{\lambda}(x_1,\dots, x_k)= \sum_{T: \ \mrm{sh}(T)=\lambda}
x^T\:,
\]
where $\mrm{sh}(T)$ is the shape of the tableau $T$. We refer the
reader to the books \cite{saga,stan2}
for more about Schur functions and their connection to the RSK 
algorithm.

Suppose that the random variables
$U=(u(i,j))_{(i,j)\in\{1,\dots,N\}\times\{1,\dots,K\}}$ are
independent and that $u(i,j)$ is geometrically distributed with
parameter $q_j$, for some fixed $q=(q_1,\ldots,q_K)\in (0,1)^K$.
Then \cite{ocon03b} the law of the random
partition $\lambda=(\lambda_1,\ldots,\lambda_K)$ associated with
$w(U)$ is given by
\begin{equation}
P\{\lambda=l\} = a(q)^N s_l(q) s_l(1,\dots ,1),\quad l\in P_K\:,
\end{equation}
where $P_K$ is the set of integer partitions with at most $K$
non-zero parts, $a(q)=\prod_j(1-q_j)$ and $s_l$ is the Schur
function associated with the integer partition $l$. In particular,
the law of $\lambda$ is {\em symmetric} in the parameters
$q_1,\ldots,q_K$. 

This extends to the level of {\em processes}: if we write
$\lambda=\lambda^{(N)}$ to emphasize its dependence on $N$, it is
easy to check~\cite{ocon03b} that the random {\em sequence}
$\lambda^{(N)}$ is a Markov chain in $P_K$ with transition
probabilities given by

\begin{equation}
P\{ \lambda^{(N+1)}=l |\ \lambda^{(N)}=m \} =
a(q)\frac{s_l(q)}{s_m(q)}\:,
\end{equation}
for all $m$ and $l$ such that $l_1\ge m_1\ge l_2\ge m_2\ge
\cdots$. In particular, we can see that the law of the sequence
$\lambda^{(N)}$ is symmetric in the parameters $q_1,\ldots,q_K$. 

These remarks give the law of the departure processes from both
of the stochastic tandem queueing and storage models described above.
Consider a saturated tandem of $K$ queues/stores as defined in \S
\ref{sse-satu}. Assume that the random variables
$u(i,j),i\in \N^*,j\in \{1,\dots, K\}$ are all independent, and that
$u(i,j)$ is geometric of parameter $q_j\in (0,1)$. Let $D^{(N)}$ be the
instant of
departure of customer $N$ from queue $K$ and let $R^{(N)}$ be the
cumulative departures over the time slots 1 to $N$ in store $K$. 
By Theorem \ref{row-queue}, $D^{(N)}$ and $R^{(N)}$ are, respectively,
the lengths of the longest and shortest rows $\lambda_1^{(N)}$ and
$\lambda_K^{(N)}$ of the random tableaux obtained by the RSK algorithm.

\paragraph{Interchangeability.}

A result by Weber \cite{webe79} (see also Anantharam \cite{anan87})
states 
that a series of ./M/1 queues arranged in tandem are interchangeable,
i.e.\ the output
process has the same distribution under any reordering of the queues.
%% This result 
%% generalized a work by Friedman \cite{frie65} for deterministic services.
This is exactly the fact that the law of $D$ is symmetric in the
parameters
and can be seen now as a direct consequence of the symmetry of the Schur
polynomials.
We can now extend this result as follows:

\begin{theorem}\label{th-reorder}
The law of $(D^{(N)},R^{(N)})_N$ is unchanged if the
parameters of the geometric laws are modified from $(q_1,\dots, q_K)$ to
$(q_{\sigma(1)}, \dots , q_{\sigma(K)})$, for any permutation $\sigma$
of $\{1,\dots , K\}$. 
\end{theorem}

In particular, stores  arranged in tandem are interchangeable. 

\paragraph{Hydrodynamic limits.}

Proving `shape theorems' (also known as `hydrodynamic limits') and
fluctuation results is an important issue in interacting particle
systems, see for instance \cite{KiLa}. 
The law of $D^{(N)}$ is well-known and the connection with RSK has
been exploited before to obtain detailed asymptotic and fluctuation
results for the corresponding interacting particle system.  The tandem
storage (and bus-stop) model however has not been discussed in the
literature. 
A by product of our analysis is that this model is amenable to some
very precise analysis and development.

Indeed, consider Theorem \ref{th-reorder} 
and the important special case when the parameters are equal, $q_j=q$ say.
Then the random variables $D$ and $R$ are jointly distributed as
the 
largest and smallest `eigenvalues' of the Meixner ensemble.  In the
exponential
case, replace Meixner by Laguerre.  For example, if the $u(i,j)$ are
exponentially
distributed with parameter $1$, then the cumulative departures $R^{(N)}$
over the
time slots 1 to N in the corresponding tandem of $K$ stores is
distributed
as the smallest eigenvalue of the $K\times K$ random matrix $AA^*$,
where $A$ is 
a $K\times N$ matrix with i.i.d. standard complex normal entries.  This
distribution is 
known, see for example~\cite{RVAl}.  When $N=K$ it is exponential
with mean $K$.  
This can be used as a starting point for the analysis of the
fluctuations and 
hydrodynamic limit corresponding to the bus-stop model of Section
\ref{sse-part}, a programme
which is beyond the scope of this paper.

\paragraph{Fixed point results.}

Finally we remark that the identity $\bigr[\min_{\pi \in \widetilde{\Pi}}
\sum_{(i,j)\in \pi} u(i,j)=R \bigl]$
can also be used as a cornerstone
for proving fixed point and related weak convergence results for stores
in tandem.  In the context of queues in tandem, these questions
have been pursued for some time and the identity 
$\bigr[ \max_{\pi \in\Pi} \sum_{(i,j)\in \pi} u(i,j)= D \bigl]$ was a key ingredient in
solving them~\cite{chan93,BBMa00,MaPr,prab03}.

\end{document}